\documentclass{article}
\usepackage{amssymb,amsmath,latexsym,amscd,amsfonts}

\newtheorem{thm}{Theorem}[section]

\newtheorem{definition}{Definition}[section]

\newenvironment{proof}{{\noindent{\bf Proof:}}}{$\hfill\Box$}

\def\tr{{\text{tr}}}
\def\dim{{\text{dim}}}

\def\v{\vert}
\def\r{\rangle}

\def\group{\mathbf{F}_{q}^+\!\rtimes\mathbf{F}_q^{\times}}
\def\near{\mathbf{H}^{+}\!\rtimes \mathbf{H}^{\times}}

\begin{document}

\title{Indistinguishable Chargeon-Fluxion Pairs in the Quantum Double of Finite Groups }

\author{Salman Beigi\\ {\it \small Institute for Quantum Information}\\ {\it \small California Institute of
Technology}\\ {\it \small Pasadena, CA}
\and Peter W. Shor \\ {\it \small Department of Mathematics}\\ {\it \small Massachusetts Institute
of Technology}\\ {\it \small Cambridge, MA}
\and Daniel Whalen\\ {\it \small Department of Mathematics}\\ {\it \small Massachusetts Institute
of Technology}\\ {\it \small Cambridge, MA}}

\maketitle

%%%%%%%%%%%%%%%%%%%%%%%%%%%%%%%%%%%%%%%%

\begin{abstract} We consider the category of finite dimensional representations of the quantum double of a finite group as a modular tensor category. We study auto-equivalences of this category whose induced permutations on the set of simple objects (particles) are of the special form of $PJ$, where $J$ sends every particle to its charge conjugation and $P$ is a transposition of a chargeon-fluxion pair. We prove that if the underlying group is the semidirect product of the additive and multiplicative groups of a finite field, then such an auto-equivalence exists. In particular, we show that for $S_3$ (the permutation group over three letters) there is a chargeon and a fluxion which are not distinguishable.
Conversely, by considering such permutations as modular invariants, we show that a transposition of a chargeon-fluxion pair forms a modular invariant if and only if the corresponding group is isomorphic to the semidirect product of the additive and multiplicative groups of a finite near-field.
\end{abstract}

\section{Introduction}

The \emph{quantum double construction} introduced by Drinfeld provides a systematic method to obtain solutions of the Yang-Baxter equation. In this construction for any finite group $G$ there is a corresponding Hopf algebra $D(G)$ which admits a universal $R$-matrix.
Drinfeld's double arises naturally in the study of $(2+1)$-dimensional quantum field theories when a continues gauge symmetry is spontaneously broken to a finite group \cite{bais1, bais2}. In this case, the excitations (particles) of the system are described by irreducible representations of $D(G)$, and their brading is given by the $R$-matrix. Also the fusion of particles can be computed by decomposing a tensor product representation into irreducible ones. More specifically, the category of finite dimensional representations of $D(G)$, denoted $\mathcal{Z}(G)$, forms a \emph{modular tensor category}, and describes the underlying algebraic structure of two-dimensional particles (anyons).

For any modular tensor category there are two matrices called the $S$-matrix and the $T$-matrix which provide a representation of the modular group $\text{SL}_2(\mathbb{Z})$:
\begin{align}
\left(
  \begin{array}{cc}
    0 & -1 \\
    1 & 0 \\
  \end{array}
\right)\mapsto S, \quad \quad  \left(
                           \begin{array}{cc}
                             1 & 1 \\
                             0 & 1 \\
                           \end{array}
                         \right) \mapsto T.
\end{align}
Matrices $M$ (with certain properties) that commute with both $S$ and $T$ are called \emph{modular invariant} and are studied extensively to classify modular invariant partition functions of a rational conformal field theory (see \cite{ostrik, davydov} and references there).

Assume that a modular tensor category admits a non-trivial auto-equivalence, that is, a permutation on simple objects (particles) under which the fusion rules and braidings remain invariant. Since matrices $S$ and $T$ are defined in terms of the category, such a permutation commutes with both of them and forms a \emph{permutational} modular invariant. In this paper we study auto-equivalences of $\mathcal{Z}(G)$. Such an auto-equivalence not only forms a modular invariant, but also means that the two labeling of particles describing by the corresponding permutation are \emph{not distinguishable}.

We consider auto-equivalences of $\mathcal{Z}(G)$ whose corresponding permutations over particles are of the following form: it first replaces every particle with its charge conjugation and then interchanges two certain particles where one of them is an electric charge (chargeon) and the other one is a magnetic flux (fluxion). We show that if $G$ is isomorphic to $\group$, the semidirect product of the additive group of a finite field and its multiplicative group, then $\mathcal{Z}(G)$ admits an auto-equivalence of the above form. Conversely, we consider these permutations as modular invariants and show that if the transposition of a chargeon-fluxion pair is a modular invariant, then the underlying group is isomorphic to $\near$ where $\mathbf{H}$ is a \emph{near-field} (near-fields will be defined consequently).

This paper is organized as follows. In the following section, after introducing some notations, we define $D(G)$ and describe its irreducible representations. We also explain the example of $G=S_3$ (the permutation group over three letters) which indeed is an example of our main result. In Section~\ref{sec:3} we first give some background from \cite{davydov} about the necessary and sufficient conditions for an equivalence between $\mathcal{Z}(G)$ and $\mathcal{Z}(G')$, for two groups $G, G'$, and then prove that there exists an auto-equivalence of $\mathcal{Z}(\group)$ of the particular type described above. In Section~\ref{sec:4} we classify all groups $G$ for which there exists a modular invariant that transposes a chargeon-fluxion pair. We also explain two non-trivial modular invariants in $A_6$ (the group of even permutations over six letters). Finial remarks and some open problems are discussed in Section~\ref{sec:conclusion}.

\section{Quantum double of a finite group}

Let us first fix some notations. $\mathbb{C}$ denotes the set of complex numbers and $\mathbb{C}^{\times}$ is its multiplicative group. $x^{\ast}$ is the complex conjugate of $x\in \mathbb{C}$ and $\v x\v^{2}=xx^{\ast}$. For a subgroup $K$ of a group $G$ and $g\in K$, $Z_K(g)$ denotes the centralizer of $g$ in $K$: $Z_{K}(g)=\{h\in K:\, hg=gh\}$. We write $g\overset{K}\sim g'$ if there exists $h\in K$ such that $hgh^{-1}=g'$. When $K=G$ and there is no confusion we drop $K$ in these notations ($Z_{G}(g)=Z(g)$ and $g\sim g'$ means $g\overset{G}\sim g'$). The conjugacy class of $g\in G$ is denoted by $\overline{g}$ ($\overline{g}=\{hgh^{-1}:\, h\in G\}$). In this paper all representations are over complex numbers, and for a representation $\rho$ of a group, $\rho^{\ast}$ denotes its complex conjugate representation. $\tr_{\rho}(\cdot)$ is the character of $\rho$, and $\mathbf{1}$ denotes the trivial representation ($\tr_{\mathbf{1}}(\cdot) =1$). $\delta$ denotes the Kronecker delta function, so for any relation $p$, $\delta_p =1$ if $p$ holds and otherwise $\delta_p = 0$. The size of a set $X$ is denoted by $\v X\v$. Finally, equivalence of categories, and isomorphism of groups and representations are shown by $\simeq$.

Let $G$ be a finite group. The quantum double or
Drinfeld double of $G$ denoted by $D(G)$, is a Hopf algebra containing $\mathbb{C}G$. $D(G)$ can be described by the $\mathbb{C}$-basis $\{gh^{\ast}: g,h \in G\}$ with the multiplication
\begin{align}
(g_1h_1^{\ast})(g_2h_2^{\ast})=\delta_{h_2=(g_2^{-1}h_1g_2)}\, (g_1 g_2)h_2^{\ast},
\end{align}
and the comultiplication
\begin{align}
\Delta(gh^{\ast}) = \sum_{h_1h_2=h} gh_1^{\ast}\otimes gh_2^{\ast}.\label{eq:comulti}
\end{align}
By $g\in D(G)$ we mean $g= \sum_h gh^{\ast}$, and $h^{\ast} = eh^{\ast}$ ($e$ is the identity of the group).
The unit of $D(G)$ is equal to $e=\sum_h e h^{\ast}$, the counit is given by $\varepsilon(g h^{\ast})=\delta_{h=e}$, and the antipode is
\begin{align}
\gamma(g h^{\ast})= g^{-1} (gh^{-1}g^{-1})^{\ast}.
\end{align}

$D(G)$ is quasi-triangular with the $R$-matrix
\begin{align}\label{eq:r-matrix}
R=\sum_{g\in G} g^{\ast}\otimes g.
\end{align}

\subsection{Representations of $D(G)$}

Consider an element $a\in G$ and let $\pi$ be a representation of $Z(a)$ over the vector space $W$ with the basis
$\{w_1, \dots, w_d\}$. Define the vector space $V_{( \overline{a}, \pi )}$ with the basis $\{\v b, w_i\r: b\in \overline{a},\, 1\leq i\leq d \}$. Then $V_{(\overline{a}, \pi)}$ is a representation of $D(G)$: for any $b\in \overline{a}$ fix $k_b\in G$ such that $b= k_b a
k_b^{-1}$. (Let $k_a = e$.) Observe that $k_{gbg^{-1}}^{-1}gk_b$ is always in $Z(a)$ and then for any $i=1,\dots , d$, $b\in
\overline{a}$, and $gh^{\ast} \in D(G)$ define
\begin{align}
g h^{\ast} \v b , w_i\r = \delta_{h= b}\, \v gbg^{-1} ,  (k_{gbg^{-1}}^{-1}gk_b)\, w_i\r.
\end{align}
$\chi_{(\overline{a}, \pi)}$, the character of this representation, is given by
\begin{align}
\chi_{(\overline{a}, \pi)} (gh^{\ast}) =  \delta_{h\in \overline{a}}\, \delta_{gh=hg}\, \tr_{\pi} (k_h^{-1} g k_h).
\end{align}

If the representation $\pi$ of $Z(a)$ is irreducible, then the representation $V_{(\overline{a}, \pi)}$ of $D(G)$ is irreducible as well. Conversely, all irreducible representations of $D(G)$ are of the above form and are indexed by conjugacy classes of $G$ and irreducible representations of the centralizer of a fixed element in the corresponding conjugacy class (see for example \cite{bakalov}). The trivial representation of $D(G)$ is indexed by $\mathbf{0}=(e, \mathbf{1})$. Moreover, the \emph{charge conjugation} of $(\overline{a}, \pi)$, which we denote by $(\overline{a}, \pi)^{\vee}$, is isomorphic to $(\overline{a^{-1}},\pi^{\ast})$.

Irreducible representations $(\overline{a}, \pi)$ are considered as particles in a two-dimensional space (anyons). $\overline{a}$ is called the magnetic flux of $(\overline{a}, \pi)$, and $\pi$ is its electric charge. A particle with the trivial magnetic flux ($a=e$) is called a \emph{chargeon}, and if $\pi=\mathbf{1}$, $(\overline{a}, \pi)$ is called a \emph{fluxion}.

\subsection{Fusion rules}

Let $(\overline{a}, \pi)$ and $(\overline{a'}, \pi')$ be two irreducible representations of $D(G)$. Then $(\overline{a}, \pi)\otimes (\overline{a'}, \pi')$ as a new representation of $D(G)$ is isomorphic to the direct sum of irreducible representations:
\begin{align}
(\overline{a}, \pi)\otimes (\overline{a'}, \pi') \simeq \bigoplus_{(\overline{h}, \rho)}
N_{(\overline{a}, \pi)(\overline{a'}, \pi')}^{(\overline{h}, \rho)} (\overline{h},
\rho).
\end{align}
Due to the comultiplication defined in~\eqref{eq:comulti}, if $N_{(\overline{a}, \pi)(\overline{a'}, \pi')}^{(\overline{h}, \rho)} $ is non-zero, $\overline{h}$ has a non-empty intersection with $\overline{a}\, \overline{a'}$, the product of the two conjugacy classes. However, to exactly compute the fusion rules we should use the {\it Verlinde formula}.

Define the matrix $S$ whose rows and columns are indexed by irreducible representations of $D(G)$ and
\begin{align}\label{eq:s-matrix}
S_{(\overline{a}, \pi)(\overline{a'}, \pi')} = \frac{1}{\v Z(a)\v \cdot \v Z(a')\v} \sum_{h:\, ha'h^{-1}\in Z(a)} \tr_{\pi}(ha'^{-1}h^{-1})\tr_{\pi'}(h^{-1}a^{-1}h).
\end{align}
Then $N_{XY}^Z$ can be computed in terms of $S$:
\begin{align}\label{eq:verlinde}
N_{XY}^{Z}=\sum_U \frac{S_{XU}S_{YU}S_{ZU}^{\ast}}{S_{\mathbf{0}U}},
\end{align}
where the summation is over all irreducible representations, and $\mathbf{0}=(e, \mathbf{1})$ is the trivial representation.

\subsection{$\mathcal{Z}(G)$}

$\mathcal{Z}(G)$ denotes the category of finite dimensional representations of $D(G)$ over complex numbers. $\mathcal{Z}(G)$ is a fusion category. Moreover, the $R$-matrix~\eqref{eq:r-matrix} defines the braiding $R_{XY}: X\otimes Y \rightarrow X\otimes Y$ for every two representations $X, Y$. As a result, $\mathcal{Z}(G)$ is a modular tensor category (see \cite{bakalov} for details).

\subsection{Example: $\mathcal{Z}(S_3)$}\label{sec:s3}

Let $G=S_3$, the permutation group over three letters: $S_3= \langle \sigma, \tau: \sigma^2=\tau^3=e, \sigma \tau = \tau^{-1}\sigma \rangle $. $D(S_3)$ has eight irreducible representations described in the following table.
\begin{align*}
\begin{array}{|c|ccc|cc|ccc|}
\hline
      & A & B & C & D & E  & F & G & H   \\
\hline
\text{conjugacy class} & e & e & e & \overline{\sigma} & \overline{\sigma} &  \overline{\tau} & \overline{\tau} & \overline{\tau}   \\
 \text{irrep of the centralizer}     & \mathbf{1} & sign & \pi & \mathbf{1} & [-1] &  \mathbf{1} & [\omega] & [\omega^{\ast}]  \\
\hline
\end{array}
\end{align*}
Here $sign$ denotes the sign representation, $\pi$ is the two-dimensional representation of $S_3$, and $[-1]$, and $[\omega], [\omega^{\ast}]$ denote the non-trivial representations of $Z(\sigma)=\{e, \sigma\}$ and $Z(\tau)=\{e, \tau, \tau^{-1}\}$. The corresponding $S$-matrix is \begin{align}\label{eq:s-s3}
S = \frac{1}{6} \left(\begin{array}{ccc|cc|cccc}
      1 & 1 & 2 &  3 &  3 &  2 & 2 & 2   \\
      1 & 1 & 2 & -3 & -3 &  2 & 2 & 2  \\
      2 & 2 & 4  & 0 & 0 &  -2 & -2 & -2  \\
      \hline
      3 & -3 & 0 &  3 & -3 &  0 & 0 & 0 \\
      3 & -3 & 0  & -3 & 3  & 0 & 0 & 0 \\
      \hline
      2 & 2 & -2  & 0 & 0  & 4 & -2 & -2 \\
      2 & 2 & -2  & 0 & 0  & -2 & -2 & 4 \\
      2 & 2 & -2  & 0 & 0  & -2 & 4 & -2
\end{array}\right),
\end{align}
and then using the Verlinde formula~\eqref{eq:verlinde} the fusion rules can be computed.

{\small
\begin{align*}
\begin{array}{|c|ccc|cc|ccc|}
\hline
  \otimes    & A & B & C  & D & E  & F & G & H   \\
\hline
 A   & {\scriptstyle A} & {\scriptstyle B} & {\scriptstyle C} & {\scriptstyle D} & {\scriptstyle E} &  {\scriptstyle F} & {\scriptstyle G} & {\scriptstyle H}   \\
 B    & {\scriptstyle B} & {\scriptstyle A} & {\scriptstyle C} & {\scriptstyle E} & {\scriptstyle D} &  {\scriptstyle F} & {\scriptstyle G} & {\scriptstyle H}  \\
 C    & {\scriptstyle C} & {\scriptstyle C} & {\scriptstyle A\oplus B\oplus C}  & {\scriptstyle D\oplus E} & {\scriptstyle D\oplus E} &  {\scriptstyle G\oplus H} & {\scriptstyle F\oplus H} & {\scriptstyle F\oplus G}  \\
\hline
 D    & {\scriptstyle D} & {\scriptstyle E} & {\scriptstyle D\oplus E} &  {\scriptstyle A\oplus C\oplus F\oplus G\oplus H} & {\scriptstyle B\oplus C\oplus F\oplus G\oplus H} &  {\scriptstyle D\oplus E} & {\scriptstyle D\oplus E} &
 {\scriptstyle D\oplus E} \\
 E    & {\scriptstyle E} & {\scriptstyle D} & {\scriptstyle D\oplus E}  & {\scriptstyle B\oplus C\oplus F\oplus G\oplus H} & {\scriptstyle A\oplus C\oplus F\oplus G\oplus H}  & {\scriptstyle D\oplus E} & {\scriptstyle D\oplus E} & {\scriptstyle D\oplus E} \\
\hline
 F    & {\scriptstyle F} & {\scriptstyle F} & {\scriptstyle G\oplus H}  & {\scriptstyle D\oplus E} & {\scriptstyle D\oplus E}  & {\scriptstyle A\oplus B\oplus F} & {\scriptstyle H\oplus C} & {\scriptstyle G\oplus C} \\
 G    & {\scriptstyle G} & {\scriptstyle G} & {\scriptstyle F\oplus H}  & {\scriptstyle D\oplus E} & {\scriptstyle D\oplus E}  & {\scriptstyle H\oplus C} & {\scriptstyle A\oplus B\oplus G} & {\scriptstyle F\oplus C} \\
 H    & {\scriptstyle H} & {\scriptstyle H} & {\scriptstyle F\oplus G}  & {\scriptstyle D\oplus E} & {\scriptstyle D\oplus E}  & {\scriptstyle G\oplus C} & {\scriptstyle F\oplus C} & {\scriptstyle A\oplus B\oplus H} \\
\hline
\end{array}
\end{align*}
}

The are several symmetries in this fusion table. In particular, by exchanging $C$ and $F$ we obtain the same table. This fact can also be seen from the $S$-matrix~\eqref{eq:s-s3}. If we let $P$ to be the permutation matrix corresponding to the transposition $(C, F)$, then $PSP^{-1}=S$. On the other hand, by the Verlinde formula the fusion rules are computed in terms of $S$, so since $S$ is invariant under $P$, the fusion rules are also symmetric with respect to the transposition $(C, F)$. In the following section we prove that this symmetry can be extended to an auto-equivalence of the whole modular tensor category.

\section{A non-trivial auto-equivalence of $\mathcal{Z}(\group)$}\label{sec:3}

In this section we show that for some class of groups $G$ there exists an auto-equivalence of $\mathcal{Z}(G)$ that sends every particle to its charge conjugation and interchanges a particular chargeon-fluxion pair.

\subsection{Equivalence of $\mathcal{Z}(G)$ and $\mathcal{Z}(G')$}

The corresponding modular tensor categories to groups $G$ and $G'$ can be equivalent even if $G$ and $G'$ are not isomorphic. Necessary and sufficient conditions for such an equivalence have been first established by Naidu and Nikshych \cite{lagrangian, naidu}. Here we present our result based on tools recently developed by Davydov \cite{davydov}. Initiated by Ostrik \cite{ostrik}, Davydov's work classifies all trivializing algebras in $\mathcal{Z}(G)$ and provides a formula to compute the corresponding characters. Using this framework one can not only find a possible equivalence between $\mathcal{Z}(G)$ and $\mathcal{Z}(G')$, but also explicitly compute the action of the corresponding equivalence on simple objects.

Before stating the main theorem regarding the equivalence of $\mathcal{Z}(G)$ and $\mathcal{Z}(G')$ we need some definitions. For a modular tensor category $\mathcal{C}$, $\mathcal{C}^{op}$ denotes the opposite category. It is known that $\mathcal{Z}(G)^{op}\simeq \mathcal{Z}(G)$ (such an equivalence will be described later).
For two categories $\mathcal{C}$ and $\mathcal{C'}$, $\mathcal{C}\boxtimes \mathcal{C'}$ is a category whose objects and morphisms are defined componentwise. $\mathcal{Z}(G)\boxtimes \mathcal{Z}(G')$ is equivalent to $\mathcal{Z}(G\times G')$ because the irreducible representations of $D(G\times G')$ are of the form $X\boxtimes X'$ where $X$ and $X'$ are irreducible representations of $D(G)$ and $D(G')$ respectively, and the $R$-matrix can be expressed componentwise as well. We refer the reader to the original paper \cite{davydov} for the definitions of trivializing algebras, the category of local modules, and the left and right parents of a given algebra.

\begin{thm} \cite{davydov, correspondence} An equivalence between two modular tensor categories $\mathcal{C}^{op}$ and $\mathcal{C'}$ corresponds to a trivializing algebra $A$ in $\mathcal{C}\boxtimes\mathcal{C'}$ such that the categories of local modules of the left and right parents of $A$ are equal to $\mathcal{C}$ and $\mathcal{C'}$ respectively.
\end{thm}

The next step toward finding an equivalence between $\mathcal{Z}(G)$ and $\mathcal{Z}(G')$ (notice that $\mathcal{Z}(G)^{op}\simeq \mathcal{Z}(G)$) is a classification of trivializing algebras in $\mathcal{Z}(G)\boxtimes \mathcal{Z}(G')\simeq \mathcal{Z}(G\times G')$.

Recall that for a group $G$, a $2$-cocycle $\varphi\in H^{2}(G, \mathbb{C}^{\times})$ is a map from $G\times G$ to non-zero complex numbers such that
\begin{align}\label{eq:cocycle}
\varphi(fg, h) \varphi(f,g) = \varphi(f, gh) \varphi(g, h).
\end{align}
For a $2$-cocycle $\varphi$ we let $\varphi(g\v h) = \varphi(g, h)\varphi(ghg^{-1}, g)^{-1}$.

\begin{thm} \cite{davydov} Trivializing algebras in $\mathcal{Z}(G)$ are in one-to-one correspondence with subgroups $K$ of $G$ and $\varphi\in H^{2}(K, \mathbb{C}^{\times})$. The character of such an algebra $A(K, \varphi)$ as a representation of $D(G)$ is given by
\begin{align}\label{eq:character}
\chi_{A(K, \varphi)}(gh^{\ast}) = \frac{1}{\v K\v}\, \delta_{gh=hg} \sum_{\substack{x:\, xgx^{-1}\in K\\ xhx^{-1}\in K}} \varphi(xgx^{-1}\v xhx^{-1}).
\end{align}
\end{thm}

According to this theorem to establish an equivalence between $\mathcal{Z}(G)^{op}\simeq \mathcal{Z}(G)$ and $\mathcal{Z}(G')$ we should find a subgroup $U\subseteq G\times G'$ and $\varphi\in H^2(U, \mathbb{C}^{\times})$ such that the categories of local modules over the parents of $A(U, \varphi)$ are the whole categories $\mathcal{Z}(G)$ and $\mathcal{Z}(G')$. Davydov \cite{davydov} has found the conditions under which the algebra $A(U, \varphi)$ satisfies the above properties, and then concluded the following theorem.

\begin{thm} \label{thm:automorphism} \cite{davydov} An equivalence between $\mathcal{Z}(G)^{op}$ and $\mathcal{Z}(G')$ corresponds to a subgroup $U\subseteq G\times G'$, and $\varphi \in H^{2}(U, \mathbb{C}^{\times})$ such that
\begin{enumerate}
\item the projections of $U$ onto the first and second components are equal to $G$ and $G'$ respectively,
\item the restriction of $\varphi(\cdot \v \cdot) $ on $(U\cap (G\times \{e\}))\times (U\cap ( \{e\}\times G'))$ is non-degenerate.
\end{enumerate}

Moreover, if such a $U$ and $\varphi$ exist, the map of the corresponding equivalence over simple objects can be computed by decomposing $\chi_{A(U, \varphi)}$ into irreducible characters of $D(G\times G')$; if ${X\boxtimes X'}$, where $X$ and $X'$ are irreducible representations of $D(G)$ and $D(G')$, appears in this decomposition then the equivalence sends $X$ to $X'$.
\end{thm}

We now give an example to clarify the above theorems.
Let $\Delta(G)=\{(g, g):\, g\in G\}$ be a subgroup of $G\times G$, and let $\varphi\in H^2(\Delta(G), \mathbb{C}^{\times})$ be the trivial cocycle ($\varphi=1$). Obviously, the projections of $\Delta(G)$ onto the first and second components are the whole group $G$. Also, since $\Delta(G)\cap (G\times \{e\}) = \Delta(G)\cap (\{e\}\times G)=\{(e, e)\}$, $\varphi(\cdot \v \cdot) $ is non-degenerate on $(\Delta(G)\cap (G\times \{e\}))\times (\Delta(G)\cap ( \{e\}\times G'))$.  As a result, $\mathcal{Z}(G)^{op}$ is equivalent to $\mathcal{Z}(G)$.

To find the action of this equivalence on simple objects we need to compute $\chi_{A(\Delta(G), 1)}$ using~\eqref{eq:character}.
For $g=(g_1, g_2)$ and $h=(h_1, h_2)$ in $G\times G$ we have
\begin{align}
\chi_{A(\Delta(G), 1)}(gh^{\ast})& = \frac{1}{\v \Delta(G)\v}\, \delta_{gh=hg}\sum_{x, y\in G} \delta_{xg_1x^{-1}=yg_2 y^{-1}}\, \delta_{xh_1x^{-1}=yh_2 y^{-1}}\\
& = \delta_{gh=hg}\, \delta_{g_1h_1^{\ast}\overset{G}\sim g_2h_2^{\ast}} \,\v Z(g_1) \cap Z(h_1) \v,
\end{align}
where by $g_1h_1^{\ast}\overset{G}\sim g_2h_2^{\ast}$ we mean there exists $l\in G$ such that $l(g_1h_1^{\ast})l^{-1} = (lg_1l^{-1})(lh_1l^{-1})^{\ast} =g_2h_2^{\ast}$.

Now define
\begin{align}
\Phi(gh^{\ast}) = \sum_{(\overline{x}, \rho)} \chi_{(\overline{x}, \rho)\boxtimes (\overline{x}, \rho^{\ast})}\, (gh^{\ast}),
\end{align}
where the summation is over all irreducible representations $(\overline{x}, \rho)$ of $D(G)$. We have
\begin{align}
\Phi(gh^{\ast}) & = \sum_{(\overline{x}, \rho)}\, \chi_{(\overline{x}, \rho)}(g_1h_1^{\ast})\, \chi_{(\overline{x}, \rho^{\ast})} (g_2h_2^{\ast}) \label{eq:phi1} \\
& = \delta_{g_1h_1 = h_1g_1} \delta_{g_2h_2 = h_2g_2} \delta_{h_1\overset{G}\sim h_2} \sum_{(\overline{h_1} , \rho)}  \tr_{\rho}(g_1)\, \tr_{\rho^{\ast}}(k_{h_2}^{-1} g_2 k_{h_2})\\
& = \delta_{gh = hg}  \delta_{h_1\overset{G}\sim h_2}
\delta_{g_1\overset{Z(h_1)}\sim  (k_{h_2}^{-1} g_2 k_{h_2}) } \v Z_{Z(h_1)}( g_1 ) \v \\
& = \delta_{gh = hg}  \delta_{g_1h_1^{\ast}\overset{G}\sim g_2h_2^{\ast}} \v Z( g_1 ) \cap Z(h_1) \v \label{eq:phi4},
\end{align}
where in the third line we use the orthogonality relations in the character table of $Z(h_1)$. As a result, $\chi_{A(\Delta(G), 1)} = \Phi$, and then the equivalence of $\mathcal{Z}(G)^{op}$ and $\mathcal{Z}(G)$ sends the simple object $(\overline{x}, \rho)$ to $(\overline{x}, \rho^{\ast})$.

\subsection{$G=\group$}

Let $\mathbf{F}_q$ be the finite field with $q$ elements and denote its additive and multiplicative group by $\mathbf{F}_q^+$ and $\mathbf{F}_q^{\times}$ respectively. Then the semidirect product of these groups is defined as follows. We represent elements of $\group$ by $(a, \alpha)$ where $a\in \mathbf{F}_q^{+}$ and $\alpha\in \mathbf{F}_q^{\times}$, and define $(a, \alpha)(a', \alpha') = (a+ \alpha \times a', \alpha \times \alpha')$ which by abuse of notation is denoted by $(a+ \alpha a', \alpha\alpha')$. The identity element of this group is $e=(0, 1)$ and the inverse of $(a, \alpha)$ is equal to $(a, \alpha)^{-1}= (-\alpha^{-1}a, \alpha^{-1})$.

We will use the following properties regarding the structure of $\group$. The conjugacy class of $(a, \alpha)$ is $\overline{(a, \alpha)}=\{(b, \alpha):\, b\in \mathbf{F}_q^{+}\}$ if $\alpha\neq 1$, and $\overline{(1, 1)}=\{(b, 1):\, b\in \mathbf{F}_q^{+}, b\neq 0\}$.
$K= \{ (a, 1):\, a\in \mathbf{F}_q^{+}  \}$ is a normal subgroup of $\group$ isomorphic to $\mathbf{F}_q^{+}$, and $K=Z(1, 1)$. Also note that for every $(a, \alpha)$ where $\alpha\neq 1$, $\v Z(a, \alpha)\v =q-1$ and $K\cap Z(a, \alpha)=\{e\}$.

Consider a non-trivial irreducible representation of $K$, and let $\pi$ be the corresponding induced representation on $\group$. Then $\tr_{\pi} (e) = q-1$, $\tr_{\pi} (1, 1) =-1$, and $\tr_{\pi} (a, \alpha)=0$ if $\alpha\neq 1$. Since $\sum_{(a, \alpha)} \v \tr_{\pi} (a, \alpha)\v^2 = \v \group\v$, $\pi$ is an irreducible representation.

Now we are ready to state and prove the main result of this section.

\begin{thm}\label{thm:main}
There exists an auto-equivalence of $\mathcal{Z}(\group)$ whose corresponding permutation on simple objects is of the form $PJ$ where $P$ is a transposition of a chargeon-fluxion pair, and $J$ sends every object to its charge conjugation ($J: X \mapsto X^{\vee} $).
\end{thm}

\begin{proof} According to Theorem~\ref{thm:automorphism} we need to introduce a subgroup $U\subseteq (\group)\times(\group)$ and $\varphi\in H^2(U, \mathbb{C}^{\times})$ which satisfy properties 1 and 2.

Let $U= \{ ((a_1, \alpha), (a_2, \alpha^{-1})):\, a_1,a_2\in \mathbf{F}_q^{+},\, \alpha\in \mathbf{F}_q^{\times} \}$. $U$ is a normal subgroup, its projection onto each component is the whole group $\group$, and $U\cap ((\group) \times\{e\}) = K\times \{e\} $ and $(\{e\}\times (\group))\cap U = \{e\}\times K$, where $K$ is defined above.

Let $p$ be the characteristic of $\mathbf{F}_q$ (so $q$ is a power of $p$), and let $\omega$ be a $p$-th root of unity ($\omega^p=1$). Also, assume that $\tr_p: \mathbf{F}_q \rightarrow \mathbf{F}_p$ is the trace function, i.e. $\tr_p(a)$ is equal to the trace of the $\mathbf{F}_p$-linear map $x\mapsto ax$. Now define $\varphi: U\times U\rightarrow \mathbb{C}^{\times}$ by
\begin{align}
\varphi( g,  h ) = \omega^{\tr_p( \alpha a_2b_1)},
\end{align}
where $g=((a_1, \alpha),( a_2, \alpha^{-1}))$ and $h=((b_1, \beta),( b_2, \beta^{-1}))$. $\varphi$ satisfies~\eqref{eq:cocycle}, and then $\varphi\in H^2(U, \mathbb{C}^{\times})$. Also, it is easy to see that $\varphi(\cdot \v \cdot )$ acts non-degenerately on $(K\times \{e\} )\times( \{e\}\times K )$. As a result, $A(U, \varphi)$ provides an equivalence of $\mathcal{Z}(\group)$ and $\mathcal{Z}(\group)^{op}$.

To obtain the permutation corresponding to this equivalence we need to compute $\chi_{A(U, \varphi)}$. Since $U$ is a normal subgroup we have
\begin{align}
\chi_{A(U, \varphi)}(gh^{\ast}) = \frac{1}{\v U\v} \delta_{gh=hg}\, \delta_{g, h \in U}\, \sum_{k} \varphi(kgk^{-1}, khk^{-1}) \varphi( khk^{-1}, kgk^{-1} )^{-1}.
\end{align}
Letting $g=((a_1, \alpha),( a_2, \alpha^{-1}))$, $h=((b_1, \beta),( b_2, \beta^{-1}))$ and $k=((x_1, \theta), (x_2, \lambda))$, we have
\begin{align}
kgk^{-1} & =  ( (x_1+\theta a_1 - \alpha x_1, \alpha), (x_2+\lambda a_2 - \alpha^{-1} x_2 , \alpha^{-1} )  ),\\
khk^{-1} & =  ( (x_1+\theta b_1 - \beta x_1, \beta), (x_2+\lambda b_2 - \beta^{-1} x_2 , \beta^{-1} )  ).
\end{align}
and thus
\begin{align}
\varphi(kgk^{-1}, khk^{-1}) & = \omega^{\tr_p( ( \alpha x_2+\alpha \lambda a_2 -x_2  )( x_1+ \theta b_1 - \beta x_1  )    )},\\
\varphi( khk^{-1}, kgk^{-1} ) & = \omega^{\tr_p(  ( \beta x_2+\beta \lambda b_2 -x_2  )( x_1+ \theta a_1 - \alpha x_1  )  )}.
\end{align}
Now observe that $gh=hg$ is equivalent to $b_1(\alpha -1) =a_1(\beta -1)$ and $\alpha(1-\beta )a_2 = \beta(1-\alpha) b_2$. So if $g$ and $h$ commute, $\varphi(kgk^{-1}, khk^{-1}) \varphi( khk^{-1}, kgk^{-1} )^{-1}$ is independent of $x_1, x_2$, and we have
\begin{align}
\chi_{A(U, \varphi)}(gh^{\ast}) = \frac{1}{\v U\v} \delta_{gh=hg}\, \delta_{g, h \in U}\, \sum_{x_1,x_2, \theta, \lambda} \omega^{\tr_p( \theta\lambda ( \alpha a_2 b_1 - \beta b_2 a_1)  )}.
\end{align}
Therefore,
\begin{align}\label{eq:27}
\chi_{A(U, \varphi)}(gh^{\ast}) =
\begin{cases}
  \delta_{gh=hg}\, \delta_{g, h \in U}\, (q-1)  & \mbox{if } \alpha b_1 a_2  = \beta b_2 a_1, \\
  - \delta_{gh=hg}\, \delta_{g, h \in U} & \mbox{if } \alpha b_1 a_2  \neq \beta b_2 a_1 .
\end{cases}
\end{align}
Notice that if either $\alpha$ or $\beta$ is not equal to $1$, $g, h\in U$ and $gh=hg$  imply $\alpha b_1 a_2  = \beta b_2 a_1$. Thus~\eqref{eq:27} can be simplified to
\begin{align}
\chi_{A(U, \varphi)}(gh^{\ast}) =
\begin{cases}
  \delta_{gh=hg}\, \delta_{g, h \in U}\, (q-1)  & \mbox{if } \alpha\neq 1 \mbox{ or } \beta\neq 1, \\
  \delta_{gh=hg}\, \delta_{g, h \in U}\, \left( \delta_{a_1 b_2 = a_2 b_1} (q-1) -  \delta_{a_1 b_2  \neq a_2 b_1}\right) & \mbox{if } \alpha = \beta =1 .
\end{cases}
\end{align}

We now need to decompose $\chi_{A(U, \varphi)}$ into irreducible characters. Let
\begin{align}
\Psi(gh^{\ast}) = \sum_{(\overline{x}, \rho)} \chi_{(\overline{x}, \rho) \boxtimes (\overline{x^{-1}}, \rho)} (gh^{\ast}).
\end{align}
By the same steps as in the computation of $\Phi(gh^{\ast})$ in~\eqref{eq:phi1}-\eqref{eq:phi4} we find that
\begin{align}
\Psi(gh^{\ast}) =  \delta_{gh=hg}\, \delta_{ g_1h_1^{\ast}\sim g_2^{-1}(h_2^{-1})^{\ast}     }\, \v Z(g_1 )\cap Z( h_1 ) \v,
\end{align}
where $g=(g_1, g_2)=( (a_1, \alpha), (a_2, \alpha'))$ and $h=(h_1, h_2)=((b_1, \beta), (b_2, \beta'))$.
Observe that if $gh=hg$ and either $\alpha\neq 1$ or $\beta\neq 1$, then $g_1h_1^{\ast}\sim g_2^{-1}(h_2^{-1})^{\ast} $ is equivalent to $g,h\in U$. This fact can be verified simply by writing these conditions in terms of $a_1, a_2, \alpha, \dots$. Also in this case $g,h\in U$ and $gh=hg$ imply $\v Z(g_1 )\cap Z( h_1 ) \v = q-1$. Moreover, if $\alpha=\beta=1$, then $ g_1h_1^{\ast}\sim g_2^{-1}(h_2^{-1})^{\ast}$ is equivalent to $a_1 b_2 = a_2 b_1$, $g_1\sim g_2^{-1}$, and $h_1\sim h_2^{-1}$. Therefore,
\begin{align}
\Psi(gh^{\ast})= \begin{cases}
  \delta_{gh=hg}\, \delta_{g, h \in U}\, (q-1)  & \mbox{if } \alpha \neq 1 \mbox{ or }\beta\neq 1, \\
  \delta_{gh=hg}\, \delta_{g_1\sim g_2^{-1}} \delta_{h_1\sim h_2^{-1}} \delta_{a_1b_2=a_2b_1}\, \v Z(g_1 )\cap Z( h_1 ) \v & \mbox{if } \alpha=\beta=1 .
\end{cases}
\end{align}

Consider the chargeon $C=( e, \pi  )$ ($\pi$ is defined above) and fluxion $F=(\overline{(1, 1)}, \mathbf{1})$ in $\mathcal{Z}(\group)$, and define
$\Gamma = \chi_{C\boxtimes C} + \chi_{F\boxtimes F} - \chi_{C\boxtimes F} - \chi_{F\boxtimes C}$. Then
\begin{align}
\Gamma(gh^{\ast}) & = \left( \chi_C(g_1h_1^{\ast}) -\chi_F(g_1h_1^{\ast}) \right) \left( \chi_C(g_2h_2^{\ast}) - \chi_F(g_2h_2^{\ast}) \right)\\
& = \delta_{gh=hg} \left(  \delta_{h_1=e}\, \tr_{\pi}(g_1) - \delta_{h_1\in \overline{(1, 1)}}  \right)\left(  \delta_{h_2=e}\, \tr_{\pi}(g_2) - \delta_{h_2\in \overline{(1, 1)}}    \right).
\end{align}

If either $\alpha\neq 1$ or $\beta\neq 1$, then $\Gamma(gh^{\ast})=0$ and we have $\chi_{A(U, \varphi)}(gh^{\ast})=\Psi(gh^{\ast})=\Psi(gh^{\ast})-\Gamma(gh^{\ast})$. Also when $\alpha=\beta=1$ by considering a few cases one can verify that $\chi_{A(U, \varphi)}(gh^{\ast}) = \Psi(gh^{\ast}) -\Gamma(gh^{\ast})$. For instance, if ($\alpha=\beta=1$ and) $a_1=0$ and $a_2\neq 0$ we have
\begin{align}
\chi_{A(U, \varphi)}(gh^{\ast}) & =\delta_{g_2h_2=h_2g_2} \delta_{\alpha'=\beta'=1} (\delta_{b_1=0}(q-1) - \delta_{b_1\neq 0})\\
& = \delta_{\alpha'=\beta'=1} (\delta_{b_1=0} (q-1) - \delta_{b_1\neq 0}),
\end{align}
and
\begin{align}
\Gamma(gh^{\ast}) &= \delta_{g_2h_2=h_2g_2}\delta_{\alpha'=\beta'=1} (\delta_{b_1=0}(q-1) - \delta_{b_1\neq 0})( -\delta_{b_2=0} - \delta_{b_2\neq 0})\\
&= \delta_{\alpha'=\beta'=1} (\delta_{b_1=0}(q-1) - \delta_{b_1\neq 0})(-1),
\end{align}
and since $g_1=e$ is not conjugate with $g_2^{-1}\neq e$, $\Psi(gh^{\ast})=0$. Thus $\chi_{A(U, \varphi)}(gh^{\ast})=\Psi(gh^{\ast})=\Psi(gh^{\ast})-\Gamma(gh^{\ast})$.

As a result, $A(U, \varphi)$ corresponds to an equivalence of $\mathcal{Z}(\group)$ and $\mathcal{Z}(\group)^{op}$ which transposes $C$ and $F$ and sends $(\overline{x}, \rho)\neq C,F$ to $(\overline{x^{-1}}, \rho)$. Combining this equivalence with the equivalence of $\mathcal{Z}(G)$ and $\mathcal{Z}(G)^{op}$ given in the previous section we obtain an auto-equivalence of $\mathcal{Z}(\group)$ whose corresponding permutation is $PJ$ where $P$ is equal to the transposition $(C, F)$. We are done.

\end{proof}

$q=2$, the simplest example of this theorem, gives the group $\mathbb{Z}_2$. There are four simple objects in $\mathcal{Z}(\mathbb{Z}_2)$ which can be described by the excitations of the toric code \cite{kitaev}. In this case the corresponding auto-equivalence is given by going to the dual lattice and interchanging Pauli-$x$ and Pauli-$z$ terms in the toric code Hamiltonian. This auto-equivalence can also be described by a domain wall \cite{liang}.

For $q=3$ the group $\group$ is isomorphic to $S_3$, and the chargeon and fluxion constructed in the proof, correspond to representations $C$ and $F$ described in Section~\ref{sec:s3}. Moreover, in $\mathcal{Z}(S_3)$ the charge conjugation of each particle is itself. As a result, this auto-equivalence of $\mathcal{Z}(S_3)$ only transposes $C$ and $F$, which means that these two particles in $\mathcal{Z}(S_3)$ are \emph{indistinguishable}.

\section{Chargeon-fluxion symmetry as a modular invariant}\label{sec:4}

The corresponding $S$-matrix to $\mathcal{Z}(G)$ is defined in~\eqref{eq:s-matrix}. The $T$-matrix is a diagonal one that contains the \emph{twist numbers} of simple objects on the diagonal. For $\mathcal{Z}(G)$, $T$ is given by
\begin{align}\label{eq:t-matrix}
T_{(\overline{g}, \pi)(\overline{g}, \pi)} = T_{(\overline{g}, \pi)} = \frac{\tr_{\pi} (g)}{\tr_{\pi}(e)}.
\end{align}
The pair of matrices $(S, T)$ is called a \emph{modular data}, and a \emph{modular invariant} corresponding to $(S, T)$ is a matrix $M$ that commutes with both $S$ and $T$, and such that all entries of $M$ are non-negative integers and $M_{\mathbf{0}\mathbf{0}}=1$ ($\mathbf{0}$ is the trivial object).
Clearly, the permutation corresponding to an auto-equivalent of a modular tensor category commutes with both $S$ and $T$ and is a modular invariant. However, a modular invariant may not even be a permutation and then may not come from an auto-equivalence.

In this section we study permutation matrices which form a modular invariant of $\mathcal{Z}(G)$. In particular, we classify all groups $G$ for which there exists a modular invariant of the form $P$ or $PJ$, where $P$ is a transposition of a chargeon-fluxion pair. Notice that $J$ always commutes with both $S$ and $T$ (it can easily be seen from the formulas~\eqref{eq:s-matrix} and~\eqref{eq:t-matrix} in the case of $\mathcal{Z}(G)$; for a proof in the general case see \cite{bakalov}). Thus, $PJ$ is a modular invariant if and only if $P$ is a modular invariant.

\subsection{Near-fields}\label{sec:4.1}

By the result of Section~\ref{sec:3}, all groups $\group$, defined in terms of a finite field, admit a transposition of a chargeon-fluxion pair as a modular invariant. Here we show that every group with a modular invariant of this form is isomorphic to $\near$ where $\mathbf{H}$ is a \emph{near-field}.

\begin{definition}
A set $\mathbf{H}$ with two binary operations $+$ and $\times$ is called a near-field if
\begin{enumerate}
\item $(\mathbf{H}, +)$ is an abelian group with the identity element $0$.

\item $0\times x = x\times 0 =0$ for every $x\in \mathbf{H}$.

\item $(\mathbf{H}\setminus 0, \times)$ is a group with the identity element $1$.

\item the multiplication is distributive from left with respect to the addition: $x\times (y + z) = x\times y + x\times z$. (Distributivity from right is not assumed.)
\end{enumerate}
\end{definition}

The class of all finite near-fields is completely known: there is a method for constructing finite near-fields due to Dickson \cite{dickson}, and it has been shown by Zassenhaus \cite{Zassenhaus} that all finite near-fields except precisely seven of them, are given by Dickson's construction.

For a near-field $\mathbf{H}$ one can consider an action of $\mathbf{H}^{\times}$ on $\mathbf{H}^+$ and define a group structure on $\near$ as follows. Elements of $\near$ are denoted by $(a, \alpha)$ where $a\in \mathbf{H}^{+}$ and $\alpha\in \mathbf{H}^{\times}$, and $(a, \alpha)(b, \beta) = (a+\alpha\times b , \alpha\times \beta)$. This multiplication turns $\near$ to a group with the identity element $e=(0, 1)$. (Notice that to obtain a group we must define the action of $\mathbf{H}^{\times}$ on $\mathbf{H}^{+}$ by multiplication from \emph{left}, and multiplication from right does not work.) $K=\{(a, 1):  a\in \mathbf{H}^+ \}$ is a subgroup of $\near$ isomorphic to $\mathbf{H}^{+}$. On the other hand, it is easy to see that all elements of $K\setminus e$ are conjugate. Thus, $K$ is an abelian group all of whose elements, except identity, have the same order. As a result, the size of this group $\v K\v= \v \mathbf{H}\v =q$ is a power of a prime number, and $K\simeq \mathbf{H}^+ \simeq \mathbf{F}_q^{+}$.

We will also use the fact that the centralizer of the multiplicative group of every near-field (with more than $2$ elements) is non-trivial. This property can be verified by checking Dickson's near-fields as well as the other seven near-fields classified by Zassenhaus (see \cite{hall}).

\subsection{A group with a chargeon-fluxion symmetry is isomorphic to $\near$}

We now state the main result of this section.

\begin{thm} \label{thm:uniqueness}
Suppose that the permutation matrix $P$ corresponding to a transposition of a chargeon-fluxion pair forms a modular invariant for $\mathcal{Z}(G)$. Then $G\simeq \near$ where $\mathbf{H}$ is a near-field. Conversely, for every group $\near$ there exists such a modular invariant.
\end{thm}

\begin{proof}
We first show that there exists a chargeon-fluxion pair in $\mathcal{Z}(\near)$ that forms a modular invariant.

Consider a non-trivial representation of the abelian subgroup $K\subseteq \near$ (defined above), and denote its induced representation on $\near$ by $\pi$. Let $C=(e, \pi)$ and $F=(\overline{a}, \mathbf{1})$, where $a=(1,1)\in \near$. We claim that the permutation $P$ which exchanges $C$ and $F$ is a modular invariant.

By the definition of $\pi$, $\dim \pi=q-1$, $\tr_{\pi}(a)=-1$ and $\tr_{\pi}(h) = 0$ for every $h\notin K$. Then since $\sum_{g} \v\tr_{\pi}(g)\v^2 =q(q-1)$, $\pi$ is an irreducible representation. Also a dimension-counting argument shows that all other irreducible representations of $\near$ come from an irreducible representation of $\mathbf{H}^{\times} \simeq (\near)/K$, and then for every such representation $\mu$, $\tr_{\mu}(a)=\tr_{\mu}(e)=\dim \mu$.

$P$ commutes with $T$ because $T_{C}=T_{F}=1$. To prove $PS=SP$ we should show that $S_{CX}=S_{FX}$ for every irreducible representation $X\neq C, F$ of $D(\near)$, and $S_{CC}=S_{FF}$. This is a straightforward computation given the structure of $Z(a)=K$ and the irreducible representations of $\near$ described above.

Now consider an arbitrary group $G$, let $C=(e, \pi)$ be a chargeon and $F=(\overline{a}, \mathbf{1})$ a fluxion in $\mathcal{Z}(G)$, and assume that the permutation $P$ which interchanges $C$ and $F$ commutes with the corresponding $S$-matrix. Notice that since by the Verlinde formula the fusion rules are computed in terms of $S$ and $PSP^{-1}=S$, the fusion rules are also symmetric with respect to $C$ and $F$. We prove $G\simeq \near$ in the following steps.\\

\noindent(a) $\pi\neq 1$ and $a \neq e$.\\

$\mathbf{0}=(e, 1)$ is the unique representation such that $\mathbf{0}\otimes X \simeq X$, so its fusion rules cannot be the same
as any other representation. Therefore, $C$ and $F$ are different from $\mathbf{0}$. $\Box$\\

\noindent(b) $\dim \pi = \v \overline{a}\v$.\\

Since $\mathbf{0}\neq C, F$ we have $S_{C\mathbf{0}}=S_{F\mathbf{0}}$. Then
\begin{align} \label{eq:dim}
\frac{\dim \pi}{\vert G\vert}=\frac{1}{\vert Z(a)\vert},
\end{align}
or equivalently $\dim \pi = \v \overline{a}\v$. $\Box$ \\

\noindent(c) $\{e\}\cup\overline{a}$ is a subgroup of $G$.\\

Let $X=(\overline{h},\mu)$ be a representation such that $\overline{h}$ is different from $\{e\}$ and $\overline{a}$. Since $C$ has a trivial magnetic flux, $C\otimes X$ is equivalent to the sum of representations whose magnetic flux is equal to $\overline{h}$. Thus, the magnetic flux of any representation in $F\otimes X$ should also be $\overline{h}$. This means that $\overline{a}\,\overline{h}=\overline{h}$ for any $h\notin \{e\}\cup\overline{a}$. As a result,  $\overline{a}\, \overline{a}\subseteq \{e\}\cup\overline{a}$ which implies that $\{e\}\cup\overline{a}$ is closed under multiplication and forms a subgroup. $\Box$\\

For every $X=(\overline{h}, \mu)$ we have
\begin{align}
S_{CX}=\frac{1}{\v G\v\cdot \v Z(h)\v}\sum_{k\in G} \tr_{\pi}(kh^{-1}k^{-1})\tr_{\mu}(e)=
\frac{\tr_{\pi}(h^{-1})\,\dim \mu\ }{\v Z(h)\v},
\end{align}
and
\begin{align}
S_{FX}= \frac{1}{\v Z(h)\v \cdot \v Z(a)\v} \sum_{khk^{-1}\in Z(a)}
\tr_{\mu}(k^{-1}a^{-1}k).
\end{align}
Therefore, if $X\neq C, F$
\begin{align}\label{eq:tr-pi-h}
\frac{\tr_{\pi}(h^{-1})\,\dim \mu\ }{\v Z(h)\v} = \frac{1}{\v Z(h)\v \cdot \v Z(a)\v} \sum_{khk^{-1}\in Z(a)}
\tr_{\mu}(k^{-1}a^{-1}k).
\end{align}

\noindent(d) For any irreducible representation $\mu$ of $G$ different from $\pi$ we have
\begin{align}
\tr_{\mu}(a^{-1})=\tr_{\mu}(a)=\dim \mu=\tr_{\mu}(e).
\end{align}

Let $h=e$ in~\eqref{eq:tr-pi-h} and note that $k^{-1}a^{-1}k$ is a conjugate of $a^{-1}$ in $Z(h)=G$. Thus $\tr_{\mu}(k^{-1}a^{-1}k)=\tr_{\mu} (a^{-1})$ and
\begin{align}
\frac{\dim \pi\,\dim \mu\ }{\v G\v} = \frac{\tr_{\mu}(a^{-1})}{\v Z(a)\v}.
\end{align}
Then by~\eqref{eq:dim} we obtain $\tr_{\mu}(a^{-1})=\dim \mu$. $\Box$\\

\noindent(e) $\tr_{\pi}(h)=0$, for any $h\notin \{e\}\cup \overline{a}$, and $\tr_{\pi}(a)=\tr_{\pi}(a^{-1})=-1$.\\

The column $\overline{h}$ of the character table of $G$ is orthogonal to columns $e$ and $\overline{a}$. On the other hand, by (d) columns
$e$ and $\overline{a}$ coincide except at the representation $\pi$. Therefore, $\tr_{\pi}(h)=0$. $\tr_{\pi}(a)=-1$ can be shown using (b), and the orthogonality of $\pi$ and the trivial representation of $G$. $\Box$\\

\noindent(f) $\v Z(a)\v =\v \overline{a}\v +1$.\\

Because of the orthogonality of columns $e$ and $\overline{a}$ of the character table of $G$ we have
\begin{align}
\sum_{\mu} \tr_{\mu}(e) \tr_{\mu}(a)^{\ast}=0,
\end{align}
where the sum is over all irreducible representations of $G$. Thus $\sum_{\mu\neq \pi} (\dim \mu)^2 - \dim \pi =0$. On the other hand, we know that
$\sum_{\mu} (\dim \mu)^2 = \v G\v$. Therefore, $\v G\v - (\dim \pi)^2 - \dim \pi =0$ which by using $\dim \pi =\v a\v $ gives $\v Z(a)\v =\v \overline{a}\v +1$. $\Box$\\

\noindent(g)  $Z(a) = \{e\} \cup \overline{a}$.\\

According to (f) it is sufficient to show that $h\notin Z(a)$ for every $h\notin \{e\}\cup \overline{a}$.
This fact can easily be seen from~\eqref{eq:tr-pi-h} by letting $\mu=\mathbf{1}$. $\Box$\\

\noindent(h) $Z(a)\simeq \mathbf{F}_q^{+}$ where $q$ is a power of a prime number, and $\v G\v = \v Z(a)\v\cdot \v \overline{a}\v = q(q-1)$.\\

Since $Z(a)=\{e\} \cup \overline{a}$ is a normal subgroup, $Z(b)=Z(a)$ for every $b\in \overline{a}$. Thus $Z(a)$ is an abelian subgroup. On the other hand, the order of all elements of $\overline{a}=Z(a)\setminus e$ is the same. Therefore, $Z(a)$ is isomorphic to $\mathbf{F}_q^+$ where $q$ is a power of a prime number. $\Box$\\

For simplicity let $Z(a)=\mathbf{H}$. Then $\mathbf{H}$ is an abelian subgroup of $G$. We show that a multiplication $\times$ can be defined on $\mathbf{H}$ which together with the operation of $\mathbf{H}$ induced from $G$ turns it into a near-field. We then prove that $G\simeq \mathbf{H}\rtimes \mathbf{H}^{\times}$.

Since $\v G/\mathbf{H}\v = \v \overline{a}\v$, and $\mathbf{H}=Z(a)$, the cosets of $G/\mathbf{H}$ are in one-to-one correspondence with elements of $\overline{a}$; for every $b\in \overline{a}$ there exists a unique $\tilde x_b= x_b\mathbf{H} \in G/\mathbf{H}$ such that $x_b ax_b^{-1}=b$. Now define a binary operation $\times$ on $\mathbf{H}$ in the following form. $e\times b = b\times e=e$ for every $b\in \mathbf{H}$, and for $b, c\in \overline{a}$
\begin{align}
b\times c = x_b x_c a x_c^{-1} x_b^{-1}.
\end{align}
$\times$ is well-defined because elements of $\mathbf{H}$ commute with every element of $\bar{a}$.\\

\noindent(i) $\mathbf{H}^{\times}= (\mathbf{H}\setminus e, \times)$ is a group whose identity element is $a$. \\

The inverse of $b$ is $b'$ where $b'=x_b^{-1}a x_b$. The associativity is proved using $\tilde x_{b\times c} = \tilde x_b \tilde x_c$. $\Box$\\

\noindent(j) $\mathbf{H}$ with the induced operation from $G$ as the addition and $\times$ as the multiplication forms a near-field.\\

We need to show that multiplication is distributive from left with respect to addition: $b\times (cd) = (b\times c)(b\times d)$. If one of $b,c,d$ is equal to $e$, it obviously holds; otherwise both sides are equal to $x_b cdx_b^{-1}$. $\Box$\\

In the following we assume that $q=\v \mathbf{H}\v >2$ since otherwise $G\simeq \mathbf{H}\rtimes \mathbf{H}^{\times}$ is obvious.\\

\noindent(k) There exists $g\in G\setminus \mathbf{H}$ such that $G=\mathbf{H}Z(g)$.\\

Since $\mathbf{H}$ is a near-field, the centralizer of $\mathbf{H}^{\times}$ is non-trivial (see Section~\ref{sec:4.1}). This means that there exists $g\in G$ such that $gag^{-1}\neq a$ and $(gag^{-1})\times b = b\times (gag^{-1})$ for every $b\in \mathbf{H}$. In other words, for every $x\in G$, $gx a x^{-1}g = xgag^{-1}x^{-1}$, or equivalently, $\overline{g}\subseteq g\mathbf{H}$. Therefore, $\v \overline{g}\v \leq \v \mathbf{H}\v =q$, and then $\v Z(g)\v \geq q-1$. On the other hand, $\mathbf{H}$ is a normal subgroup of $G$, so $\mathbf{H}Z(g)$ is a subgroup and since $Z(g)\cap \mathbf{H}=Z(g)\cap Z(a)=\{e\}$, the size of this subgroup is equal to $q\v Z(g)\v$. Thus $\v Z(g)\v \leq q-1$ and therefore, $Z(g)$ is a subgroup of order $q-1$ and $G= \mathbf{H}Z(g)$. $\Box$\\

\noindent(l) $G\simeq \mathbf{H}\rtimes \mathbf{H}^{\times}$.\\

Since $G=\mathbf{H}Z(g)$ and $\mathbf{H}\cap Z(g)=\{e\}$, every element of $G$ can uniquely be written in the form of $bk$ where $b\in \mathbf{H}$ and $k\in Z(g)$. It is easy to see that the map which sends $bk\in G$ to $(b, kak^{-1})\in \mathbf{H}\rtimes \mathbf{H}^{\times}$ is an isomorphism. We are done.

\end{proof}

\subsection{Example: two modular invariants in $\mathcal{Z}(A_6)$}\label{sec:a6}

Assume that the transposition $(X, Y)$ forms a modular invariant in $\mathcal{Z}(G)$. Theorem~\ref{thm:uniqueness} classifies all groups for which there exists such a modular invariant when $X$ is a chargeon and $Y$ is a fluxion. If we relax this assumption by keeping $X$ to be a chargeon but assuming $Y=(\overline{a}, \rho)$ is arbitrary, most steps in the proof of Theorem~\ref{thm:uniqueness} (with some variations) still hold. In particular, $\dim \rho =1$  is enough to show that $G\simeq \near$.  (In this case proving (g) needs more work.)

There are two remaining cases. First, both $X$ and $Y$ are chargeon, and second, non of them is chargeon. The first case cannot happen; if $X=(e, \pi)$ and $Y=(e, \pi')$, $S_{\mathbf{0}X}=S_{\mathbf{0}Y}$ implies $\dim \pi=\dim \pi'$. Also, for every $g\neq e $, $S_{X(\overline{g}, \mathbf{1})}=S_{Y(\overline{g},\mathbf{1})}$ is equivalent to $\tr_{\pi}(g)=\tr_{\pi'}(g)$. Thus $\pi=\pi'$.

Now assume that $X=(\overline{a}, \rho)$ and $Y=(\overline{b}, \rho')$, and $a,b\neq e$. Then for every irreducible representation $\pi$ of $G$, $S_{X(e, \pi)} = S_{Y(e, \pi)}$, and we obtain
\begin{align}
\frac{\tr_{\pi}(a^{-1})\dim \rho}{\v Z(a)\v} = \frac{\tr_{\pi}(b^{-1})\dim \rho'}{\v Z(b)\v}.
\end{align}
For $\pi=\mathbf{1}$ we find that $\dim \rho/\v Z(a)\v = \dim \rho' /\v Z(b)\v$, and thus for every $\pi$, $\tr_{\pi}(a^{-1})=\tr_{\pi}(b^{-1})$. Equivalently, $a$ and $b$ belong to the same conjugacy class, and $X,Y$ have the same magnetic flux.

Here we present an example of a modular invariant in the latter case ($X=(\overline{a}, \rho)$ and $Y=(\overline{a}, \rho')$). Let $A_6$ be the alternating group of order six (the group of even permutations over $\{1, \dots, 6\}$).
Let $a=(1, 2)(3, 4)$. $\overline{a}$ is equal to the set of all permutations of the form $(t_1, t_2)(t_3, t_4)$. Then $\v \overline{a}\v = 45$, and $\v Z(a)\v= \v A_6\v/ \v \overline{a} \v = 8$;
\begin{align}
Z(a)=\{ e, a, b_1=(1,2)(5,6), b_2=(3,4)(5,6), b_3=(1,3)(2,4), b_4=(1,4)(2,3),\nonumber\\ c_1=(1,3,2, 4)(5,6), c_2=(1,4,2,3)(5,6) \}.
\end{align}
The conjugacy classes of $Z(a)$ are $\{e\}$, $\{a\}$, $\{b_1, b_2\}$, $\{b_3, b_4\}$, and $\{c_1, c_2\}$, and the character table of $Z(a)$ is as follows.
\begin{align}
\begin{array}{|c|ccccc|}
\hline
       &  e   &   a    &    b_1,b_2   &    b_3,b_4   &    c_1,c_2    \\
\hline
\rho_1 &  1   &   1    &      1       &       1      &       1        \\
\rho_2 &  1   &   1    &      -1      &       -1     &       1         \\
\hline
\rho_3 &  1   &   1    &      1      &       -1     &       -1         \\
\rho_4 &  1   &   1    &      -1      &       1     &       -1         \\
\hline
\mu    &  2   &   -2   &      0      &       0     &       0         \\
\hline
\end{array}
\end{align}
We claim that both transpositions $(X_1, X_2)$ and $(X_3, X_4)$, where $X_i=(\overline{a}, \rho_i)$, are modular invariants.

$T$ is invariant under these transpositions because $T_{X_i}=1$ for every $1\leq i\leq 4$.

For every $Y=(e, \pi)$, $S_{X_iY} = \tr_{\pi}(a^{-1})\dim \rho_i/\v Z(a)\v$, and since $\dim \rho_i =1$ for every $i$, $S_{X_iY}=S_{X_jY}$ for $i,j\in \{1, 2, 3, 4\}$.

For every $Y=(\overline{h}, \pi)$, where $h\notin \{e\}\cup \overline{a}$, we have
\begin{align}
S_{X_iY} = \frac{1}{\v Z(a)\v\cdot \v Z(h)\v} \sum_{k:\, khk^{-1}\in Z(a)} \tr_{\rho_i}(kh^{-1}k^{-1}) \tr_{\pi}(k^{-1}a^{-1}k).
\end{align}
Observe that $c_1, c_2$ are the only elements of $Z(a)$ which can be conjugates of $h^{-1}$, and $\tr_{\rho_i} (c_1) =\tr_{\rho_j}(c_1)$ for $i,j\in \{1,2\}$ and $i,j\in \{3,4\}$. Therefore, $S_{X_iY}=S_{X_jY}$.

Now it remains to show that $S_{X_1X_1}=S_{X_2X_2}$, $S_{X_3X_3}=S_{X_4X_4}$, $S_{X_1X_3}=S_{X_2X_3}=S_{X_2X_4}=S_{X_1X_4}$, and $S_{X_1Y}=S_{X_2Y}=S_{X_3Y}=S_{X_4Y}$, where $Y=(\overline{a}, \mu)$. These equalities can simply be verified given the character table of $Z(a)$.

\section{Conclusion}\label{sec:conclusion}

We have shown that there are non-trivial symmetries in the fusion rule and braiding of particles in the quantum double model of groups of the form $\group$. In particular, there are two particles in $\mathcal{Z}(S_3)$, one of which is a chargeon and the other is a fluxion, which are indistinguishable. Conversely, we have proved that such a symmetry may appear only if the corresponding group is isomorphic to $\near$ for a near-field $\mathbf{H}$.

Some questions arise naturally from our results. First, can we extend the modular invariants described in Theorem~\ref{thm:uniqueness} to an auto-equivalence of $\mathcal{Z}(\near)$? Following the proof of Theorem~\ref{thm:main}, the first difficulty is to define the subgroup $U$; the same subgroup as in Theorem~\ref{thm:main} does not work in general because $\mathbf{H}^{\times}$ is not abelian. Indeed, it is not hard to see that if such a subgroup $U$ exists, then there must be a group automorphism $f: \mathbf{H}^{\times} \rightarrow \mathbf{H}^{\times}$ such that $\alpha$ and $f(\alpha^{-1})$ are conjugate. However, solely from this assumption, we cannot conclude that $\mathbf{H}^{\times}$ is abelian because, for example, all the inner automorphisms of the permutation group $S_n$ have this property.

Another problem is to extend Theorem~\ref{thm:uniqueness} to the case where $F=(\bar{a}, \rho)$ is an arbitrary particle. As mentioned in Section~\ref{sec:a6}, $\dim \rho=1$ again gives $G\simeq \near$. Either proving the same result in the general case or disproving by giving a counterexample is of interest.

Finally, we do not know to what extend these results can be generalized to the twisted quantum double of $\group$ or $\near$.\\

\noindent{\bf Acknowledgements.} We gratefully acknowledge Alexei Kitaev for many clarifications and useful comments. We are also thankful to Liang kong for bringing the work of Ostrik to our attention, and Chris Heunen and Alexei Davydov for many clarifications about notations. SB is partially supported by NSF under Grant No. PHY-0803371 and by NSA/ARO under Grant No. W911NF-09-1-0442.


\small
\begin{thebibliography}{10}
\vspace*{.04in}

\bibitem{bais1} A. F. Bais, P. van Driel, and P. de Wild, Anyons in discrete gauge theories with Chern-Simons terms, Nuclear physics. B 393, (1993) 547-570.

\bibitem{bais2} A. F. Bais, B. J. Schroers, and J. K. Slingerland, Hopf symmetry breaking and confinement in (2+1)-dimensional gauge theory, Journal of High Energy Physics 05 (2003) 068.


\bibitem{bakalov} B. Bakalov and A. Kirillov Jr, Lectures on tensor categories and modular functors. University Lecture Series, 21. American Mathematical Society, Providence, RI.


\bibitem{davydov} Alexei Davydov, Modular invariants for group-theoreticl modular data I,
Journal of Algebra, Volume 323, Issue 5, (2010), 1321-1348




\bibitem{dickson} L.E. Dickson, On finite algebras, Nachr. Akad. Wiss. G\"ottingen, Math.
Phys. Kl. II (1905) 358-393.



\bibitem{correspondence} J. Fr\"ohlich, J. Fuchs, I. Runkel, and C. Schweigert. Correspondences of Ribbon Categories. Adv. Math. 199 (2006), pp. 192ñ329.

\bibitem{hall} Marshall Hall, The theory of groups, New York, MacMillan Co., 1959.


\bibitem{kitaev} A. Yu. Kitaev, Fault-tolerant quantum computation by anyons, Ann. of Phys. 303 (2003).

\bibitem{naidu} Deepak Naidu, Categorical Morita Equivalence for Group-Theoretical Categories, Comm. Algebra, 35: 3544-3565, 2007.

\bibitem{lagrangian} Deepak Naidu and Dmitri Nikshych, Lagrangian Subcategories and Braided Tensor Equivalences of Twisted Quantum Doubles of Finite Groups, Comm.  Math. Physi. 279, 845-872 (2008).

\bibitem{ostrik} Viktor Ostrik, Module Categories over the Drinfeld Double of a Finite Group, Int. Math. Res. Not. 27 (2003), 1507-1520.


\bibitem{Zassenhaus} H. Zassenhaus, \"Uber endliche Fastk\"orper, Abh. Math. Sem. Univ. Hamburg, 11 (1936), 187-220


\bibitem{liang} Liang Kong, Alexei Kitaev, personal communication (2009).


\end{thebibliography}
\end{document}